\documentclass[12pt]{article}
\usepackage{graphicx}
\usepackage{amssymb}

\usepackage[latin1]{inputenc}

\newcommand{\newc}{\newcommand}
\newc{\beq}{\begin{equation}}
\newc{\eeq}{\end{equation}}
\newc{\bea}{\begin{array}}
\newc{\eea}{\end{array}}
\newcommand{\ben}{\begin{eqnarray}}
\newcommand{\een}{\end{eqnarray}}
\newc{\ra}{\rightarrow}
\newc{\bfx}{{\bf x}}
\newc{\bfV}{{\bf V}}
\newc{\cO}{{\cal O}}
\newc{\bfv}{{\bf v}}
\newc{\bfu}{{\bf u}}
\newc{\bfp}{{\bf p}}
\newc{\ve}{{\varepsilon}}
\newc{\Psibar}{\overline\Psi}
\newc{\w}{{\bf w}}
\newc{\E}{{\mathbf{E}}}
\newc{\EE}{{\mathcal E}}
\newc{\bfn}{{\mathbf\nabla}}
\newc{\la}{{\cal L}}
\newc{\tla}{{\tilde{\cal L}}}
\newc{\bp}{{\bf p}}
\newc{\ho}{\hookrightarrow }
\newc{\bP}{{\bf P}}
\newc{\pd}{{\partial}}
\newc{\piv}{{\partial_4}}
\newc{\pv}{{\partial_5}}
\newc{\bJ}{{\bf J}}
\newc{\bze}{{\mathbf 0}}
\newc{\bK}{{\bf K}}
\newc{\tphi}{{\tilde\phi}}
\newc{\tF}{{\tilde F}}
\newc{\tD}{{\tilde D}}
\newc{\tJ}{{\tilde J}}
\newc{\tj}{{\tilde j}}
\newc{\bD}{{\bf D}}
\newc{\tvphi}{{\tilde\varphi}}
\newc{\trho}{{\tilde\rho}}
\newc{\ttheta}{{\tilde\theta}}
\newc{\tpsi}{{\tilde\psi}}
\newc{\tu}{{\tilde u}}
\newc{\cD}{{\cal D}}
\newc{\tPhi}{{\tilde\Phi}}
\newc{\tPsi}{{\tilde\Psi}}
\newc{\tA}{{\tilde A}}
\newc{\talpha}{{\tilde\alpha}}
\newc{\tbeta}{{\tilde\beta}}
\newc{\bA}{{\mathbf A}}
\newc{\bB}{{\bf B}}
\newc{\br}{{\bf r}}
\newc{\sig}{{\mathbf\sigma}}
\newc{\eg}{{\rm e.g.\ }}
\newc{\ie}{{\rm i.e.\ }}
\newcommand{\bey}{\begin{eqnarray}}
\newcommand{\pslash}{\not{\hbox{\kern-2.3pt $p$}}}
\newcommand{\pdslash}{\not{\hbox{\kern-2pt $\partial$}}}
\newcommand{\eey}{\end{eqnarray}}
\newtheorem{definition}{Definition}
\newtheorem{theorem}{Theorem}
\newtheorem{proposition}[theorem]{Proposition}
\newenvironment{proof}[1][Proof]{\noindent\textbf{#1.} }{\ \rule{0.5em}{0.5em}}

\begin{document}

\begin{titlepage}
\vskip 2cm
\begin{center}
{\Large  Thermal Lie Superalgebras
\footnote{{\tt matrindade@uneb.br}}}
 \vskip 10pt
{  Marco A. S. Trindade$^{1}$, Eric Pinto$^{2}$   \\}
\vskip 5pt
{\sl $^1$Departamento de Ciências Exatas e da Terra, Universidade do Estado da Bahia,
Colegiado de Física, Bahia, Brazil \\
\sl $^2$Instituto de Física, Universidade Federal da Bahia,
 Bahia, Brazil}

\vskip 2pt
\end{center}

\begin{abstract}

We derive a general formulation for thermal Lie superalgebras motivated by thermofield dynamics formalism (TFD). Particularly, we construct the thermal Poincaré superalgebras. The operators in TFD are defined through the doubling of the degrees of freedom of the system and it can be related to Hopf algebras. In this way we explore the notion of quantum group associated with these superalgebras and we show the non-commutativity in this thermal scenario. Furthermore, the thermal M-superalgebra is also derived from TFD prescription.

\end{abstract}

\bigskip

{\it Keywords:} Lie superalgebras, thermofield dynamics.


\vskip 3pt

\end{titlepage}


\newpage

\setcounter{footnote}{0} \setcounter{page}{1} \setcounter{section}{0} %
\setcounter{subsection}{0} \setcounter{subsubsection}{0}
\section{Introduction}
\label{intro}
Supersymmetric superalgebras are pivotal tools in the construction of a supersymmetric theory, relating bosons and fermions \cite{Wess,Freund,Sohnius,Sal}. Supersymmetry is described by Lie superalgebras \cite{Kac} that contains generators in the spinor representation of the space-time symmetry group \cite{Traubenberg}. Of particular interest is the Poincaré SUSY algebra, a non trivial extension of Poincaré algebra. Specifically, it is a $Z_{2}$-graded vector space equipped with a graded Lie bracket in which the even part contains the Poincaré algebra and the odd part corresponds to spinors. Poincaré SUSY algebra is crucial in superstring theory and M-theory \cite{Kaku,Toppan1,Toppan2,Townsend,Holten}.
\  \

In the superstring theory, a thermofield dynamics approach (TFD) has been developed to construct thermal superstring and to calculate their thermodynamical properties \cite{Nedel}. In this way, thermal effects on the supersymmetry break can be analyzed. TFD is a real operator field theory at finite temperature \cite{Khanna}. Thermal field theories can be analyzed via w*-algebras resulting in the doubled Tomita-Takesaki representation related to structural elements of thermofield dynamics \cite{Matos}. Let $\pi (\mathcal{A})$ be a Tomita-Takesaki representation of w*-algebra $\mathcal{A}$ on a Hilbert space $\mathfrak{H}$ and $\sigma: \mathfrak{H} \rightarrow \mathfrak{H}$ antilinear isometry with $\sigma^{2}=1$ \cite{Matos,Tak,Ojima,Bratteli}. Thus \cite{Matos}
\begin{eqnarray}
\sigma \pi_{\omega} (\mathcal{A}) \sigma = \widetilde{\pi}_{\omega} (\mathcal{A}),
\end{eqnarray}
where $\omega$ is a state (GNS construction \cite{Bratteli}) on the w*-algebra $\mathcal{A}$, define an *-anti-isomorphism such that
\begin{eqnarray}
\left[\pi_{\omega} (\mathcal{A}), \widetilde{\pi}_{\omega}(\mathcal{A}) \right]=0.
\end{eqnarray}
In TFD formalism we have a doubling of Hilbert space $\mathfrak{H} \otimes \widetilde{\mathfrak{H}}$. The operators $O_{i}$ on the $\mathfrak{H} \otimes \widetilde{\mathfrak{H}}$ must satisfy the following tilde conjugation rules $({\sim})$ \cite{Khanna}:

\begin{eqnarray}
(O_{i}O_{j})^{\sim}&=&\widetilde{O}_{i}\widetilde{O}_{j}, \nonumber \\
(O_{i}+ \alpha O_{j})^{\sim}&=&\widetilde{O_{i}}+\alpha^{*}\widetilde{O}_{j}, \nonumber \\
(O^{\dag})^{\sim}&=&(\widetilde{O})^{\dag}, \nonumber \\
\widetilde{\widetilde{O}}&=&O, \nonumber \\
\left[O_{i},\widetilde{O}_{j}\right]&=&0. \label{tilde}
\end{eqnarray}

This leads a notion of thermal Lie algebra \cite{Santana} related to aforementioned w*-algebra. The meaning of Lie algebras was discussed in \cite{Khanna}. Particularly, the thermal Lie-Poincaré algebra was derived in \cite{Matos}.
\  \

In this work we will explore the notion of thermal algebras in the context of supersymmetry. We perform a general formulation for the thermal superalgebras. Specifically, we show how to construct such algebras and their representations. Supersymmetric algebras in four and eleven dimensions will be investigated in a thermal scenario. In addition we will also investigate structure of Hopf algebras \cite{Majid,Pre,Burban,Resh} in four dimensions. A motivation is the interest in noncommutative spaces \cite{Kobayashi} related to brane solutions in M-theory \cite{Nair,Taylor}.
\ \

The paper is organized as follows. Section 2 contains a general formulation about thermal Lie superalgebras. In section 3 we consider the thermal Poincaré superalgebra, its associated Hopf algebra and deformed (anti)commutation relations. In Section 3 we presented representations of thermal M-superalgebra from hat operators and Section 4 contains the conclusions.

\section{Thermal superalgebras}

A fundamental structures in TFD prescription are the thermal Lie algebras. Based in references \cite{Khanna,Santana} we present the next two definitions in order to derive our results. We will be using the Einstein summation convention in this manuscript.
\begin{definition}
A doubled Lie algebra $\mathfrak{g_{D}}$ is a Lie algebra with a Lie product
\begin{eqnarray}
\left[A_{i}, A_{j}\right]&=&C_{ij}^{k}A_{k}, \nonumber \\
\left[\widetilde{A}_{i}, \widetilde{A}_{j}\right]&=&-C_{ij}^{k}\widetilde{A}_{k}, \nonumber \\
\left[\widetilde{A}_{i},{A_{j}}\right]&=&0.
\end{eqnarray}
\end{definition}

\begin{definition}
A thermal Lie algebra $\mathfrak{g_{T}}$ is a Lie algebra with a Lie product given by
\begin{eqnarray}
\left[A_{i},A_{j}\right]&=&C_{ij}^{k}A_{k}, \nonumber \\
\left[\widehat{A_{i}},\widehat{A_{j}}\right]&=&C_{ij}^{k}\widehat{A_{k}}, \nonumber \\
\left[\widehat{A_{i}},A_{j}\right]&=&C_{ij}^{k}A_{k}.
\end{eqnarray}
\end{definition}
With definition of hat operators $\widehat{A}=A-\widetilde{A}$, it is easy verify that the above relations are satisfied \cite{Khanna,Santana}. Now, we present our results.
\begin{proposition}
Let $\mathfrak{h}$ be a direct sum of two isomorphic Lie algebras $\mathfrak{g}$ and $\mathfrak{\widetilde{g}}$, i.e.,  $\mathfrak{h}=\mathfrak{g} \oplus \mathfrak{\widetilde{g}}$, where $\mathfrak{g} \simeq \mathfrak{\widetilde{g}}$. Then there is two subalgebras of $\mathfrak{h}$, $\mathfrak{h_{T}}, \mathfrak{h_{D}}\subset \mathfrak{h}$,  isomorphic to algebras $\mathfrak{g_{T}}$ and $\mathfrak{g_{D}}$.
\end{proposition}

\begin{proof}
Consider the elements $B_{i}^{\mathfrak{g}}=(X_{i},0)$, $B_{i}^{\mathfrak{\widetilde{g}}}=(0, \widetilde{X}_{i})$ and $\widehat{B}_{i}=(X_{i},0)-(0,\widetilde{X}_{i})$ . Consequently,
\begin{eqnarray}
[\widehat{B}_{i}, \widehat{B}_{j}]&=&[(X_{i}, -\widetilde{X}_{i}), (X_{j}, -\widetilde{X}_{j})] \nonumber \\
&=&([X_{i}, X_{j}], [-\widetilde{X}_{i}, -\widetilde{X}_{j}]) \nonumber \\
&=&C_{i,j}^{k}(X_{k}, - \widetilde{X}_{k})  \nonumber \\
&=&C_{i,j}^{k}\widehat{B}_{k}.
\end{eqnarray}
Analogously,
\begin{eqnarray}
[\widehat{B}_{i}, B_{j}^{\mathfrak{g}}]=([X_{i},X_{j}],[-\widetilde{X}_{i},0])=C_{ij}^{k}(X_{k},0)=C_{ij}^{k} B_{k}^{\mathfrak{g}}
\end{eqnarray}
and
\begin{eqnarray}
[B_{i}^{\mathfrak{g}}, B_{j}^{\mathfrak{g}}]=C_{ij}^{k} B_{k}^{\mathfrak{g}}.
\end{eqnarray}
Thus $h_{\mathfrak{T}}\equiv span\{B_{j}^{\mathfrak{g}},\widehat{B_{j}}\}$ is a subalgebra of $\mathfrak{h}$. Now, we will show that $\psi: span\{B_{j}^{\mathfrak{g}},\widetilde{B}_{j}\}\subset\mathfrak{g}\oplus \mathfrak{\widetilde{\mathfrak{g}}}\rightarrow \mathfrak{g}_{\mathfrak{T}}$ given by $\psi(B_{j}^{\mathfrak{g}})=A_{j}$ and $\psi(\widehat{B}_{j})=\widehat{A}_{j}$ is an homomorphism. In fact,
\begin{eqnarray}
\psi[(X_{i},-X_{i}),(X_{j},-X_{j})]=\psi[\widehat{B}_{i},\widehat{B}_{j}]&=&\psi(C_{ij}^{k}\widehat{B}_{k}) \nonumber \\
&=&C_{ij}^{k}\psi(\widehat{B}_{k}) \nonumber \\
&=&[\widehat{A}_{i}, \widehat{A}_{j}] \nonumber \\
&=&[\psi(\widehat{B}_{i}), \psi(\widehat{B}_{j})].
\end{eqnarray}
Similarly, $\varphi(\widehat{A}_{j})=\widehat{B}_{j}$ also defines a homomorphism. It is easy to verify that
\begin{eqnarray}
\psi \circ \varphi = id_{\mathfrak{h_{T}}},
\varphi \circ \psi = id_{\mathfrak{g_{T}}},
\end{eqnarray}
where $id_{\mathfrak{h_{T}}}(\widehat{B}_{i})=\widehat{B}_{i}$ and $id_{\mathfrak{g_{T}}}(\widehat{A}_{i})=\widehat{A}_{i}$. Therefore $\psi=\varphi^{-1}$ and we have an isomorphism. We can perform an analogous procedure for $h_{\mathfrak{D}}\equiv span\{B_{j}^{\mathfrak{g}},\widetilde{B}_{j}\}$ and the proof is finished.
\end{proof}

\begin{definition}
The doubled Lie superalgebra is a superalgebra $G^{\mathfrak{D}}=G_{0}^{\mathfrak{D}}\oplus G_{1}^{\mathfrak{D}}$ with a product $[ \ \ , \ \ ]$ satisfying the following axioms
\begin{eqnarray}
[X_{i}, X_{j}]&=&-(-1)^{(degX_{i})(degX_{j})}[X_{j}, X_{i}], \nonumber
 \end{eqnarray}
\begin{eqnarray}
[\widetilde{X}_{i}, \widetilde{X}_{j}]=-(-1)^{(deg\widetilde{X}_{i})(deg\widetilde{X}_{j})}[\widetilde{X}_{j}, \widetilde{X}_{i}], \nonumber
\end{eqnarray}
\begin{eqnarray}
[\widetilde{X}_{i}, X_{j}]&=&-(-1)^{(deg\widetilde{X}_{i})(degX_{j})}[X_{j}, \widetilde{X}_{i}], \nonumber
\end{eqnarray}
\begin{eqnarray}
[X_{i},[X_{j},X_{k}]]=[[X_{i},X_{j}],X_{k}]+(-1)^{(degX_{i})(degX_{j})}[X_{j},[X_{i},X_{k}]], \nonumber
\end{eqnarray}
\begin{eqnarray}
[\widetilde{X}_{i},[\widetilde{X}_{j},\widetilde{X}_{k}]]=[[\widetilde{X}_{i},\widetilde{X}_{j}],\widetilde{X}_{k}]+(-1)^{(deg\widetilde{X}_{i})(deg\widetilde{X}_{j})}[\widetilde{X}_{j},[\widetilde{X}_{i},\widetilde{X}_{k}]], \nonumber
\end{eqnarray}
\begin{eqnarray}
[X_{i}^{\mathfrak{D}},[X_{j}^{\mathfrak{D}},X_{k}^{\mathfrak{D}}]]=[[X_{i}^{\mathfrak{D}},X_{j}^{\mathfrak{D}}],X_{k}^{\mathfrak{D}}]+(-1)^{(degX_{i}^{\mathfrak{D}})(degX_{j}^{\mathfrak{D}})}[X_{j}^{\mathfrak{D}},[X_{i}^{\mathfrak{D}},X_{k}^{\mathfrak{D}}]], \nonumber
\end{eqnarray}
\begin{eqnarray}
\left[X_{i}, X_{j}\right]&=&C_{ij}^{k}X_{k}, \nonumber \\
\left[\widetilde{X}_{i}, \widetilde{X}_{j}\right]&=&-C_{ij}^{k}\widetilde{X}_{k}, \nonumber \\
\left[\widetilde{X}_{i},{X_{j}}\right]&=&0,
\end{eqnarray}
where $(deg X_{i})$ denotes the degree of $X_{i}$ and $X_{i}^{\mathfrak{D}}=X_{i}$ or $\widetilde{X}_{i}$ so that $[X_{i}^{\mathfrak{D}},[X_{j}^{\mathfrak{D}},X_{k}^{\mathfrak{D}}]]=[\widetilde{X}_{i},[X_{j},X_{k}]]$ or $[X_{i},[\widetilde{X}_{j},X_{k}]]$ or $[\widetilde{X}_{i},[X_{j},\widetilde{X}_{k}]],...$.
\end{definition}

\begin{definition}
The thermal Lie superalgebra is a superalgebra $G^{\mathfrak{T}}=G_{0}^{\mathfrak{T}}\oplus G_{1}^{\mathfrak{T}}$ with a product $[ \ \ , \ \ ]$ satisfying the following axioms
\begin{eqnarray}
[X_{i}, X_{j}]&=&-(-1)^{(degX_{i})(degX_{j})}[X_{j}, X_{i}], \nonumber
 \end{eqnarray}
\begin{eqnarray}
[\widehat{X}_{i}, \widehat{X}_{j}]=-(-1)^{(deg\widehat{X}_{i})(deg\widehat{X}_{j})}[\widehat{X}_{j}, \widehat{X}_{i}], \nonumber
\end{eqnarray}
\begin{eqnarray}
[\widehat{X}_{i}, X_{j}]&=&-(-1)^{(deg\widehat{X}_{i})(degX_{j})}[X_{j}, \widehat{X}_{i}], \nonumber
 \end{eqnarray}
\begin{eqnarray}
[X_{i},[X_{j},X_{k}]]=[[X_{i},X_{j}],X_{k}]+(-1)^{(degX_{i})(degX_{j})}[X_{j},[X_{i},X_{k}]], \nonumber
\end{eqnarray}
\begin{eqnarray}
[\widehat{X}_{i},[\widehat{X}_{j},\widehat{X}_{k}]]=[[\widehat{X}_{i},\widehat{X}_{j}],\widehat{X}_{k}]+(-1)^{(deg\widehat{X}_{i})(deg\widehat{X}_{j})}[\widehat{X}_{j},[\widehat{X}_{i},\widehat{X}_{k}]], \nonumber
\end{eqnarray}
\begin{eqnarray}
[X_{i}^{\mathfrak{T}},[X_{j}^{\mathfrak{T}},X_{k}^{\mathfrak{T}}]]=[[X_{i}^{\mathfrak{T}},X_{j}^{\mathfrak{T}}],X_{k}^{\mathfrak{T}}]+(-1)^{(degX_{i}^{\mathfrak{T}})(degX_{j}^{\mathfrak{T}})}[X_{j}^{\mathfrak{T}},[X_{i}^{\mathfrak{T}},X_{k}^{\mathfrak{T}}]], \nonumber
\end{eqnarray}
\begin{eqnarray}
\left[X_{i}, X_{j}\right]&=&C_{ij}^{k}X_{k}, \nonumber \\
\left[\widehat{X}_{i}, \widehat{X}_{j}\right]&=&C_{ij}^{k}\widehat{X}_{k}, \nonumber \\
\left[\widehat{X}_{i},{X_{j}}\right]&=&C_{ij}^{k}X_{k},
\end{eqnarray}
where $(deg X_{i})$ denotes the degree of $X_{i}$ and $X_{i}^{\mathfrak{T}}=\widehat{X}_{i}$ or $\widetilde{X}_{i}$ so that $[X_{i}^{\mathfrak{T}},[X_{j}^{\mathfrak{T}},X_{k}^{\mathfrak{T}}]]=[\widehat{X}_{i},[X_{j},X_{k}]]$ or $[X_{i},[\widehat{X}_{j},X_{k}]]$ or $[\widehat{X}_{i},[X_{j},\widehat{X}_{k}]],...$.
\end{definition}

\begin{theorem}
Let $H=G \oplus \widetilde{G}$ be a direct sum of two isomorphic Lie superalgebras $G$ and $\widetilde{G}$. Then there is two subalgberas $H_{\mathfrak{T}}$ and $H_{\mathfrak{D}}$ isomorphic to superalgebras $G_{\mathfrak{T}}$ and $G_{\mathfrak{D}}$, respectively.
\end{theorem}

\begin{proof}
As $G$ and $\widetilde{G}$ are Lie superalgebras $H=G_{0} \oplus G_{1} \oplus \widetilde{G}_{0} \oplus \widetilde{G}_{1}$. We set $B_{ij}=(Y_{i},\widetilde{Y}_{j})\in H$, with $Y_{i} \in G_{0} \oplus G_{1}$ and $\widetilde{Y_{j}} \in \widetilde{G}_{0} \oplus \widetilde{G}_{1}$. Thus the direct sum is defined in the usual way \cite{Kac}:
\begin{eqnarray}
[B_{ij},B_{kl}]&=&[(Y_{i}, \widetilde{Y}_{j}),(Y_{k}, \widetilde{Y}_{l})] \nonumber \\
&=&\left([Y_{i},Y_{k}], [\widetilde{Y}_{j}, \widetilde{Y}_{l}]\right) \nonumber \\
&=&\left(-(-1)^{(deg Y_{i})(deg Y_{k})}[Y_{k},Y_{i}],-(-1)^{(deg \widetilde{Y}_{l})(deg \widetilde{Y}_{j})}[\widetilde{Y}_{j}, \widetilde{Y}_{l}]\right).
\end{eqnarray}
If we set $\widehat{B}_{i}=(Y_{i}, -\widetilde{Y}_{i})$, $B_{i}^{G}=(Y_{i},0)$, $B_{i}^{\widetilde{G}}=(0, -\widetilde{X}_{i})$, following the same steps as the proof of Proposition 1, we have that $H_{\mathfrak{T}}=span\{\widehat{B}_{i},B_{i}^{G}\}$ and $H_{\mathfrak{D}}=span\{B_{i}^{G},B_{i}^{\widetilde{G}}\}$ are isomorphic to thermal superalgebras $G^{\mathfrak{T}}$ and $G^{\mathfrak{D}}$, respectively, with the isomorphism $\Psi: H_{\mathfrak{T}} \rightarrow G^{\mathfrak{T}}$ given by $\Psi(\widehat{B}_{i})=\widehat{X}_{i}$, $\Psi(B_{i}^{G})=X_{i}$ and the isomorphism $\Phi: H_{\mathfrak{D}} \rightarrow G^{\mathfrak{D}}$ defined by  $\Phi(\widehat{B}_{i})=\widehat{X}_{i}$ and  $\Phi(B_{i}^{\widetilde{G}})=\widetilde{X}_{i}$.
\end{proof}

\begin{theorem}
Let  $U(\mathfrak{h})$ be an universal enveloping superalgebra of Lie superalgebra $\mathfrak{h}$ and $\rho(U(\mathfrak{h}))\rightarrow \mathfrak{gl}(V)$ a representation of $U(\mathfrak{h})$ on $V$. Then $\rho$ induces a representation of an universal enveloping superalgebra $U(\mathfrak{h_{T}})$ of thermal Lie superalgebra $\mathfrak{h_{T}}$, denoted by $\rho^{'}(U(\mathfrak{h_{T}}))$, with $U(\mathfrak{h_{T}})=\frac{T(\mathfrak{h_{T}})}{I}$, where $T(\mathfrak{h_{T}})$ is the tensor superalgebra over the space $\mathfrak{h_{T}}$ and $I$ is the ideal generated by $[X_{i}^{\mathfrak{T}},X_{j}^{\mathfrak{T}}]-X_{i}^{\mathfrak{T}} \otimes X_{j}^{\mathfrak{T}} + (-1)^{(deg X_{i}^{\mathfrak{T}})(deg X_{j}^{\mathfrak{T}})} X_{j}^{\mathfrak{T}} \otimes X_{i}^{\mathfrak{T}}$, ($X_{i}^{\mathfrak{T}}=X_{i}$ or $X_{i}^{\mathfrak{T}}=\widehat{X}_{i}$).
\end{theorem}

\begin{proof}
We define
\begin{eqnarray}
\rho^{'}(\widehat{X_{i}})=\rho(X_{i}) \otimes 1 + 1 \otimes \rho(X_{i})
\end{eqnarray}
and
\begin{eqnarray}
\rho^{'} (X_{i})= \rho(X_{i}) \otimes 1.
\end{eqnarray}
We consider $(a_{i} \otimes b_{i})(a_{j} \otimes b_{j})=(-1)^{(deg b_{i}) (deg b_{j})}(a_{i}a_{j} \otimes b_{i}b_{j})$ if $(deg X_{i})=(deg X_{j})=1$ and $(a_{i} \otimes b_{i})(a_{j} \otimes b_{j})=(a_{i}a_{j} \otimes b_{i}b_{j})$, otherwise. Thus in the first case $[X_{i},X_{j}]=[X_{j},X_{i}]\equiv \{X_{i},X_{j}\} \equiv X_{i},X_{j}+X_{j},X_{i}$. Therefore
\begin{eqnarray}
\{\rho(\widehat{X}_{i}), \rho(\widehat{X}_{j})\}&=&\rho(X_{i})\rho(X_{j}) \otimes 1 + \rho(X_{i}) \otimes \rho (X_{j}) - \rho(X_{j}) \otimes \rho(X_{i})+ 1 \otimes \rho(X_{j}) \rho(X_{i}) \nonumber \\
&+&\rho(X_{j})\rho(X_{i}) \otimes 1 + \rho(X_{j}) \otimes \rho (X_{i}) - \rho(X_{i}) \otimes \rho(X_{j})+ 1 \otimes \rho(X_{i}) \rho(X_{j}) \nonumber \\
&=&C_{ij}^{k}(\rho(X_{k})\otimes 1+ 1 \otimes \rho(X_{k}))= C_{ij}^{k} \rho(\widehat{X}_{k}). \nonumber
\end{eqnarray}
Moreover $\{\rho(\widehat{X}_{i}), \rho(X_{j})\}=C_{ij}^{k}\rho(X_{k})$, $\{\rho(X_{i}), \rho(X_{j})\}=C_{ij}^{k}\rho(X_{k})$, $[\rho(\widehat{X}_{i}), \rho(\widehat{X}_{j})]=C_{ij}^{k}\rho(\widehat{X}_{k})$, $[\rho(\widehat{X}_{i}), \rho(X_{j})]=C_{ij}^{k}\rho(X_{k})$ and $[\rho(X_{i}), \rho(X_{j})]=C_{ij}^{k}\rho(X_{k})$. This complets the proof of theorem.
\end{proof}

\section{Thermal Lie-Poincaré superalgebra and non-commutativity}
 We review the the the Thermal Lie Poincaré algebra and the Poincaré superalgebra in order to introduce the thermal Poincaré superalgebra. The thermal Lie-Poincaré algebra is given by \cite{Khanna,Santana}:

\begin{eqnarray}
\left[P_{\mu},P_{\nu}\right]&=&0, \nonumber \\
\left[M_{\mu\nu}, M_{\sigma \rho}\right]&=&-i(\eta_{\mu \sigma}M_{\nu\sigma}-\eta_{\nu \rho}M_{\mu\sigma}+\eta_{\mu \sigma}M_{\rho\nu}-\eta_{\nu \sigma}M_{\rho\nu}), \nonumber \\
\left[M_{\mu \nu}, P_{\sigma}\right]&=&i(\eta_{\mu\sigma}P_{\mu}-\eta_{\sigma,\mu}P_{\nu}), \nonumber \\
\left[\widehat{P}_{\mu},\widehat{P}_{\nu}\right]&=&0, \nonumber \\
\left[\widehat{M}_{\mu\nu}, \widehat{M}_{\sigma \rho}\right]&=&-i(\eta_{\mu \sigma}\widehat{M}_{\nu\sigma}-\eta_{\nu \rho}\widehat{M}_{\mu\sigma}+\eta_{\mu \sigma}\widehat{M}_{\rho\nu}-\eta_{\nu \sigma}\widehat{M}_{\rho\nu}), \nonumber \\
\left[\widehat{M}^{\mu\nu},\widehat{P}_{\sigma}\right]&=&i(\eta_{\mu\sigma}\widehat{P}_{\mu}-\eta_{\sigma,\mu}\widehat{P}_{\nu}), \nonumber \\
\left[\widehat{P}_{\mu},P_{\nu}\right]&=&0, \nonumber \\
\left[\widehat{M}_{\mu\nu}, M_{\sigma \rho}\right]&=&-i(\eta_{\mu \sigma}M_{\nu\sigma}-\eta_{\nu \rho}M_{\mu\sigma}+\eta_{\mu \sigma}M_{\rho\nu}-\eta_{\nu \sigma}M_{\rho\nu}, \nonumber \\
\left[\widehat{M}^{\mu\nu},P_{\sigma}\right]&=&i(\eta_{\mu\sigma}P_{\mu}-\eta_{\sigma,\mu}P_{\nu}),
\end{eqnarray}
where $P_{\mu}$ stands for the generators of translation, $M_{\mu\nu}$ are generators of rotations, $\eta_{\mu\nu}$ is such that $diag(\eta_{\mu\nu})=(1,-1,-1,-1)$ and $\eta_{\mu \nu}=0$ for $\mu \neq \nu$, where $\mu, \nu=0,1,2,3$.
The Poincaré superalgebra is given by \cite{Wess}

\begin{eqnarray}
\{Q_{a},Q_{b}\}&=&-2(\gamma^{\mu}C)_{ab}P_{\mu}, \nonumber \\
\{\overline{Q}_{a},\overline{Q}_{b}\}&=&-2(C^{-1}\gamma^{\mu})_{ab}P_{\mu}, \nonumber \\
\{Q_{a},\overline{Q}_{b}\}&=&2(\gamma^{\mu})_{ab}P_{\mu}, \nonumber \\
\left[P_{\mu},Q_{a}\right]&=&0, \nonumber \\
\left[M_{\mu\nu}, Q_{a}\right]&=&-(\sigma_{\mu\nu}^{4})_{ab}Q_{b}, \nonumber \\
\left[P_{\mu},P_{\nu}\right]&=&0, \nonumber \\
\left[M_{\mu\nu}, M_{\rho \sigma}\right]&=&-i(\eta_{\mu \rho}M_{\nu\rho}-\eta_{\nu \sigma}M_{\mu\rho}+\eta_{\mu \rho}M_{\sigma\nu}-\eta_{\nu \rho}M_{\sigma\nu}), \nonumber \\
\left[M_{\mu\nu},P_{\rho}\right]&=&-i(\eta_{\mu\rho}P_{\mu}-\eta_{\rho,\mu}P_{\nu}), \label{sP}
\end{eqnarray}
where $\sigma_{\mu\nu}^{4}=\frac{i}{4}\left[\gamma_{\mu},\gamma_{\nu}\right]$ with indices $a$, $b$ run from 1 to 4, $\gamma_{\mu}$ are Dirac matrices, $Q_{a}$ are Majorana spinors (spinor charge) and $C$ is the charge conjugation matrix.

Based on the previous prescription (Theorem 1) and in the definition $\widehat{Q}=Q-\widetilde{Q}$ for the thermal Majorana spinors we can derive the following additional relations to Poincaré superalgebra (\ref{sP}), obtaining the thermal superalgebra
\begin{eqnarray}
\{\widehat{Q}_{a},\widehat{Q}_{b}\}&=&-2(\gamma^{\mu}C)_{ab}\widehat{P}_{\mu}, \nonumber \\
\{\widehat{\overline{Q}}_{a},\widehat{\overline{Q}}_{b}\}&=&-2(C^{-1}\gamma^{\mu})_{ab}\widehat{P}_{\mu}, \nonumber \\
\{\widehat{Q}_{a},\widehat{\overline{Q}}_{b}\}&=&2(\gamma^{\mu})_{ab}\widehat{P}_{\mu}, \nonumber \\
\left[\widehat{P}_{\mu},\widehat{P}_{\nu}\right]&=&0, \nonumber \\
\left[\widehat{M}_{\mu\nu}, \widehat{M}_{\rho \sigma}\right]&=&-i(\eta_{\mu \rho}\widehat{M}_{\nu\rho}-\eta_{\nu \sigma}\widehat{M}_{\mu\rho}+\eta_{\mu \rho}\widehat{M}_{\sigma\nu}-\eta_{\nu \rho}\widehat{M}_{\sigma\nu}), \nonumber \\
\left[\widehat{M}_{\mu\nu},\widehat{P}_{\rho}\right]&=&-i(\eta_{\mu\rho}\widehat{P}_{\mu}-\eta_{\rho,\mu}\widehat{P}_{\nu}), \nonumber \\
\{\widehat{Q}_{a},Q_{b}\}&=&-2(\gamma^{\mu}C)_{ab}P_{\mu}, \nonumber \\
\left[\widehat{P}_{\mu},Q_{a}\right]&=&0 \nonumber, \\
\left[\widehat{M}_{\mu\nu}, Q_{a}\right]&=&-(\sigma_{\mu\nu}^{4})_{ab}Q_{b}, \nonumber \\
\left[\widehat{M}_{\mu\nu}, M_{\rho \sigma}\right]&=&-i(\eta_{\mu \rho}M_{\nu\rho}-\eta_{\nu \sigma}M_{\mu\rho}+\eta_{\mu \rho}M_{\sigma\nu}-\eta_{\nu \rho}M_{\sigma\nu}), \nonumber \\
\left[\widehat{M}_{\mu\nu},P_{\rho}\right]&=&-i(\eta_{\mu\rho}P_{\mu}-\eta_{\rho,\mu}P_{\nu}).
\end{eqnarray}
It is possible to express the thermal supersymmetric extension of the Poincaré algebra in two-component Weyl formulation resulting in
\begin{eqnarray}
\{Q_{\alpha},\overline{Q}_{\dot{\beta}}\}&=&2\sigma^{\mu}_{\alpha\dot{\beta}}P_{\mu}, \nonumber \\
\left[P_{\mu}, Q_{\alpha}\right]&=&0, \nonumber \\
\left[P_{\mu}, \overline{Q}_{\dot{\alpha}}\right]&=&0, \nonumber \\
\left[M_{\mu\nu},Q^{\dot{\alpha}}\right]&=&i(\sigma_{\mu\nu})_{\alpha}^{\beta}Q_{\beta}, \nonumber \\
\left[M_{\mu\nu},\overline{Q}^{\dot{\alpha}}\right]&=&i(\sigma_{\mu\nu})_{\alpha}^{\beta}Q_{\beta}, \nonumber \\
\{\widehat{Q}_{\alpha},\widehat{\overline{Q}}_{\dot{\beta}}\}&=&2\sigma^{\mu}_{\alpha\dot{\beta}}\widehat{P}_{\mu}, \nonumber \\
\left[\widehat{P}_{\mu}, \widehat{Q}_{\alpha}\right]&=&0, \nonumber \\
\left[\widehat{P}_{\mu}, \widehat{\overline{Q}}_{\dot{\alpha}}\right]&=&0, \nonumber \\
\left[\widehat{M}_{\mu\nu},\widehat{Q}^{\dot{\alpha}}\right]&=&i(\sigma_{\mu\nu})_{\alpha}^{\beta}\widehat{Q}_{\beta}, \nonumber \\
\left[\widehat{M}_{\mu\nu},\widehat{\overline{Q}}^{\dot{\alpha}}\right]&=&i(\sigma_{\mu\nu})_{\alpha}^{\beta}\widehat{Q}_{\beta}, \nonumber \\
\{\widehat{Q}_{\alpha},\overline{Q}_{\dot{\beta}}\}&=&2\sigma^{\mu}_{\alpha\dot{\beta}}P_{\mu}, \nonumber \\
\left[\widehat{P}_{\mu}, Q_{\alpha}\right]&=&0, \nonumber \\
\left[\widehat{P}_{\mu}, \overline{Q}_{\dot{\alpha}}\right]&=&0, \nonumber \\
\left[\widehat{M}_{\mu\nu},Q^{\dot{\alpha}}\right]&=&i(\sigma_{\mu\nu})_{\alpha}^{\beta}Q_{\beta}, \nonumber \\
\left[\widehat{M}_{\mu\nu},\overline{Q}^{\dot{\alpha}}\right]&=&i(\sigma_{\mu\nu})_{\alpha}^{\beta}Q_{\beta}.
\end{eqnarray}
The other relations are analogous to the previous case. Analogously to the reference \cite{Kobayashi} we can construct a $Z_{2}$-graded Hopf algebra $H$ with the graded tensor product $(a \otimes b)(c \otimes d)=(-1)^{(deg b) (deg c)}(ac \otimes bd$), where $a,b,c,d \in H$ and $(deg \ a)=0$ if $a$ is fermionic and $(deg \ a)=1$, if $a$ is bosonic. The universal enveloping thermal Poincaré superalgebra becomes a $Z_{2}$-Hopf algebra with the following definitions: the multiplication $m(x \otimes y)=xy$, $m(\widehat{x} \otimes \widehat{y})=\widehat{x}\widehat{y}$; $m(\widehat{x} \otimes y)=\widehat{x}y$, coproduct $\Delta(x)=x \otimes 1 + 1 \otimes x$, $\Delta(\widehat{x})=\widehat{x} \otimes 1 + 1 \otimes \widehat{x}$; counit $\epsilon(x)=\epsilon(\widehat{x})=0$; $\epsilon(1)=1$ and antipode $S(x)=-x$, $S(\widehat{x})=-\widehat{x}$, $S(1)=1$. For example,
\begin{eqnarray}
\Delta(\{\widehat{Q}_{\alpha},\overline{Q}_{\dot{\beta}}\})&=&2 \sigma_{\alpha \dot{\beta}}(P_{\mu}\otimes1+1 \otimes P_{\mu}) \nonumber \\
&=&\Delta(2\sigma_{\alpha \dot{\beta}}P_{\mu}) \nonumber \\
&=&\Delta(\widehat{Q}_{\alpha})\Delta(\overline{Q}_{\dot{\beta}})+\Delta(\overline{Q}_{\dot{\beta}})\Delta(\widehat{Q}_{\alpha}).
\end{eqnarray}
Now we consider the thermal twist elements:
\begin{eqnarray}
\mathfrak{F}^{T}=\exp\left[-\frac{1}{2}K^{\alpha \beta}(Q_{\alpha}\otimes Q_{\beta}+\widehat{Q}_{\alpha}\otimes \widehat{Q}_{\beta})\right],
\end{eqnarray}
where $K^{\alpha \beta}$, that satisfies the conditions:
\begin{eqnarray}
\mathfrak{F}_{12}^{T}(\Delta \otimes id)(\mathfrak{F}^{T})&=&\mathfrak{F}_{23}^{T}(id \otimes \Delta)(\mathfrak{F}^{T}) \nonumber \\
(\epsilon \otimes id) (\mathfrak{F}^{T})&=&(id \otimes \epsilon)(\mathfrak{F}^{T}).
\end{eqnarray}

After the twisting we have the deformed multiplication
\begin{eqnarray}
m_{\mathfrak{F}^{T}}(v_{T} \otimes w_{T}) &\equiv& v_{T} \star w_{T} \nonumber \\
&=&m\left((\mathfrak{F}^{T})^{-1}\triangleright(v_{T} \otimes w_{T})\right),
\end{eqnarray}
considering $v_{T},w_{T} \in V_{T}\equiv V \oplus V_{T}$ a left module of H and $\triangleright$ denotes the action on $V_{T}$. We can check that
\begin{eqnarray}
h\triangleright(m_{\mathfrak{F}^{T}}(v_{T} \otimes w_{T}))&=&h\triangleright(m_{\mathfrak{F}^{T}}((v+\hat{v})\otimes(w+\hat{w})) \nonumber \\
&=&h\triangleright(m_{\mathfrak{F}^{T}}(v \otimes w))+h\triangleright(m_{\mathfrak{F}^{T}}(v \otimes \widehat{w})) \nonumber \\
&+h&\triangleright(m_{\mathfrak{F}^{T}}(\widehat{v} \otimes w))+h\triangleright(m_{\mathfrak{F}^{T}}(\widehat{v} \otimes \widehat{w})) \nonumber \\
&=&m_{\mathfrak{F}^{T}}(\Delta(h)(\mathfrak{F}^{T})^{-1}\triangleright (v \otimes w)) +m_{\mathfrak{F}^{T}}(\Delta(h)(\mathfrak{F}^{T})^{-1}\triangleright (v \otimes \widehat{w})) \nonumber \\
&+&m_{\mathfrak{F}^{T}}(\Delta(h)(\mathfrak{F}^{T})^{-1}\triangleright (\widehat{v} \otimes w)) +m_{\mathfrak{F}^{T}}(\Delta(h)(\mathfrak{F}^{T})^{-1}\triangleright (\widehat{v} \otimes \widehat{w})) \nonumber \\
&=&m_{\mathfrak{F}^{T}}(\Delta^{\mathfrak{F}^{T}}(h)\triangleright (v \otimes w)) +m_{\mathfrak{F}^{T}}(\Delta^{\mathfrak{F}^{T}}(h)\triangleright (v \otimes \widehat{w})) \nonumber \\
&+&m_{\mathfrak{F}^{T}}(\Delta^{\mathfrak{F}^{T}}(h)\triangleright (\widehat{v} \otimes w)) +m_{\mathfrak{F}^{T}}(\Delta^{\mathfrak{F}^{T}}(h)\triangleright (\widehat{v} \otimes \widehat{w})),
\end{eqnarray}
where $\Delta^{\mathfrak{F}^{T}}(h)=\mathfrak{F}^{T}\Delta(h)(\mathfrak{F}^{T})^{-1}$. Note that $\mathfrak{F}^{T}$ may take the form $\mathfrak{F}^{T}=f^{\alpha}\otimes f_{\alpha}+\widehat{f}^{\alpha}\otimes \widehat{f}_{\alpha}$

Thus we have that
\begin{eqnarray}
\theta^{\alpha} \star \theta^{\beta}&=&m_{\mathfrak{F}^{T}}(\theta^{\alpha} \otimes \theta^{\beta}) \nonumber \\
&=&2 \theta_{\alpha}\theta_{\beta}+K^{\alpha \beta}.
\end{eqnarray}
Therefore
\begin{eqnarray}
\{\theta^{\alpha},\theta^{\beta}\}_{\star}&=&2K^{\alpha \beta}.
\end{eqnarray}
Analogously
\begin{eqnarray}
\left[x^{\mu},x^{\nu}\right]_{\star}&=&K^{\alpha \beta}\sigma_{\alpha \dot{\gamma}}^{\mu}\sigma_{\beta \dot{\delta}}^{\mu}
(\overline{\theta}^{\gamma}\overline{\theta}^{\delta}+\widehat{\overline{\theta}}^{\gamma}\widehat{\overline{\theta}}^{\delta})\nonumber \\
&=&\left[\widehat{x}^{\mu},x^{\nu}\right]_{\star}=\left[\widehat{x}^{\mu},\widehat{x}^{\nu}\right]_{\star}, \nonumber \\
\left[x^{\mu},\theta^{\alpha}\right]_{\star}&=&-2iK^{\alpha \beta}\sigma_{\beta \dot{\gamma}}^{\mu}(\overline{\theta}^{\gamma}+ \widehat{\overline{\theta}}^{\dot{\gamma}})\nonumber \\
&=&\left[\widehat{x}^{\mu},\theta^{\alpha}\right]_{\star}=\left[\widehat{x}^{\mu},\widehat{\theta}^{\alpha}\right]_{\star} =\left[x^{\mu},\widehat{\theta}^{\alpha}\right]_{\star}.
\end{eqnarray}
where we used that
\begin{eqnarray}
Q_{\alpha}&=&i\frac{\partial}{\partial \theta_{\alpha}}-\sigma_{\alpha \dot{\beta}}^{\mu}\overline{\theta}^{\beta}\partial_{\mu},  \nonumber \\
\widehat{Q}_{\alpha}&=&i\frac{\partial}{\partial \widehat{\theta}_{\alpha}}-\sigma_{\alpha \dot{\beta}}^{\mu}\widehat{\overline{\theta}}^{\beta}\partial_{\widehat{\mu}}.
\end{eqnarray}

Therefore, we have a deformation of the (anti)commutation relations for the coordinates.
\section{Thermal M-superalgebra}
In M-theory, the M-superalgebra is given by \cite{Toppan1,Toppan2,Townsend}
\begin{eqnarray}
\{Q'_{r},Q'_{s}\}=(C \Gamma_{\mu})_{rs}P'^{\mu}+(C \Gamma_{[ \mu \nu ]})_{rs}Z'^{\mu \nu}+(C \Gamma_{[ \mu_{1}...\mu_{5} ]})_{rs}Z'^{\mu_{1}...\mu_{5}},
\end{eqnarray}
where $Z'^{\mu \nu}$, $Z'^{\mu_{1}...\mu_{5}}$ are tensorial central charges related to p-branes and
\begin{eqnarray}
\Gamma_{[\mu_{1}...\mu_{t}]}=\frac{1}{t!}\sum_{\sigma \in S_{p}} \epsilon (\sigma) \Gamma_{\mu_{\sigma(1)}}...\Gamma_{\mu_{\sigma(t)}},
\end{eqnarray}
with $ \epsilon (\sigma)$ stands for the signature of permutation $\sigma$ and $S_{p}$ is the symmetric group.
In order to perform TFD prescription and following the same steps of Theorem 2, we explicitly present a realization of the tilde operators:
\begin{eqnarray}
Q_{r}& \equiv &\rho(Q'_{r}) \otimes 1, \ \ \widetilde{Q}_{r} \equiv 1 \otimes \rho(Q'_{r}), \nonumber \\
Q_{s}& \equiv &\rho(Q'_{s}) \otimes 1, \ \ \widetilde{Q}_{b} \equiv 1 \otimes \rho(Q_{\alpha}), \nonumber \\
P^{\mu}& \equiv & \rho(P'^{\mu}) \otimes 1, \ \  \widetilde{P}^{\mu} \equiv 1 \otimes \rho(P'^{\mu}), \nonumber \\
Z^{\mu \nu}& \equiv &\rho(Z'^{\mu \nu}) \otimes 1, \ \ \widetilde{Z}^{\mu \nu} \equiv 1 \otimes \rho(Z'^{\mu \nu}), \nonumber \\
Z^{\mu_{1}...\mu_{5}}& \equiv &\rho(Z'^{\mu_{1}...\mu_{5}}) \otimes 1, \ \ \widetilde{Z}'^{\mu_{1}...\mu_{5}} \equiv 1 \otimes \rho(Z'^{\mu_{1}...\mu_{5}})
\end{eqnarray}
so that
\begin{eqnarray}
\{Q_{r},Q_{s}\}&=&(C \Gamma_{\mu})_{rs}P^{\mu}+(C \Gamma_{[ \mu \nu ]})_{rs}Z^{\mu \nu}+(C \Gamma_{[ \mu_{1}...\mu_{5} ]})_{rs}Z^{\mu_{1}...\mu_{5}}, \nonumber \\
\{\widetilde{Q}_{r},\widetilde{Q}_{s}\}&=&(C \Gamma_{\mu})_{rs}\widetilde{P}^{\mu}+(C \Gamma_{[ \mu \nu ]})_{rs}\widetilde{Z}^{\mu \nu}+(C \Gamma_{[ \mu_{1}...\mu_{5} ]})_{rs}\widetilde{Z}^{\mu_{1}...\mu_{5}}, \nonumber \\
\{\widetilde{Q}_{r},Q_{s}\}&=&0.
\end{eqnarray}
Taking account that
\begin{eqnarray}
\widehat{Q}_{r}& \equiv &\rho(Q'_{r}) \otimes 1 + 1 \otimes \rho(Q'_{r}),  \nonumber \\
\widehat{Q}_{s}& \equiv &\rho(Q'_{s}) \otimes 1 + 1 \otimes \rho(Q'_{s}),  \nonumber \\
\widehat{P}^{\mu}& \equiv & \rho(P'^{\mu}) \otimes 1 + 1 \otimes \rho(P'^{\mu}) ,  \nonumber \\
\widehat{Z}^{\mu \nu}& \equiv &\rho(Z'^{\mu \nu}) \otimes 1 + 1 \otimes \rho(Z'^{\mu \nu}) ,  \nonumber \\
\widehat{Z}^{\mu_{1}...\mu_{5}}& \equiv &\rho(Z'^{\mu_{1}...\mu_{5}}) \otimes 1 + 1 \otimes \rho(Z'^{\mu_{1}...\mu_{5}}),
\end{eqnarray}

in according to Theorem 2, we obtain a representation of thermal M-superalgebra
\begin{eqnarray}
\{Q_{r},Q_{s}\}&=&(C \Gamma_{\mu})_{rs}P^{\mu}+(C \Gamma_{[ \mu \nu ]})_{rs}Z^{\mu \nu}+(C \Gamma_{[ \mu_{1}...\mu_{5} ]})_{rs}Z'^{\mu_{1}...\mu_{5}}, \nonumber \\
\{\widehat{Q}_{r},\widehat{Q}_{s}\}&=&(C \Gamma_{\mu})_{rs}\widehat{P}^{\mu}+(C \Gamma_{[ \mu \nu ]})_{rs}\widehat{Z}^{\mu \nu}+(C \Gamma_{[ \mu_{1}...\mu_{5} ]})_{rs}\widehat{Z}^{\mu_{1}...\mu_{5}}, \nonumber \\
\{\widehat{Q}_{r},Q_{s}\}&=&(C \Gamma_{\mu})_{rs}P^{\mu}+(C \Gamma_{[ \mu \nu ]})_{rs}Z^{\mu \nu}+(C \Gamma_{[ \mu_{1}...\mu_{5} ]})_{rs}Z^{\mu_{1}...\mu_{5}}.
\end{eqnarray}
Analogously following \cite{Toppan1,Toppan2} we can construct the tilde operators of octonionic M-superalgebra:
\begin{eqnarray}
\{Q_{r},Q_{s}\}&=&\{Q_{r}^{\ast},Q_{s}^{\ast}\}=0, \nonumber \\
\{Q_{r},Q_{s}^{\ast}\}&=&P^{\mu}(C\Gamma_{\mu})_{rs}+Z_{O}^{\mu \nu}(C\Gamma_{\mu \nu})_{rs}, \nonumber \\
\{\widetilde{Q}_{r},\widetilde{Q}_{s}\}&=&\{\widetilde{Q}_{r}^{\ast},\widetilde{Q}_{s}^{\ast}\}=0, \nonumber \\
\{\widetilde{Q}_{r},\widetilde{Q}_{s}^{\ast}\}&=&\widetilde{P}^{\mu}(C\Gamma_{\mu})_{rs}+\widetilde{Z}_{O}^{\mu \nu}(C\Gamma_{\mu \nu})_{rs}, \nonumber \\
\{\widetilde{Q}_{r},Q_{s}^{\ast}\}&=&0,
\end{eqnarray}
where $\ast$ is the conjugation of octonions. Then we have a thermal octonionic M-superalgebra
\begin{eqnarray}
\{Q_{r},Q_{s}\}&=&\{Q_{r}^{\ast},Q_{s}^{\ast}\}=0 \nonumber \\
\{Q_{r},Q_{s}^{\ast}\}&=&P^{\mu}(C\Gamma_{\mu})_{rs}+Z_{O}^{\mu \nu}(C\Gamma_{\mu \nu})_{rs}, \nonumber \\
\{\widehat{Q}_{r},\widehat{Q}_{s}\}&=&\{\widehat{Q}_{r}^{\ast},\widehat{Q}_{s}^{\ast}\}=0, \nonumber \\
\{\widehat{Q}_{r},\widehat{Q}_{s}^{\ast}\}&=&\widehat{P}^{\mu}(C\Gamma_{\mu})_{rs}+\widehat{Z}_{O}^{\mu \nu}(C\Gamma_{\mu \nu})_{rs}, \nonumber \\
\{\widehat{Q}_{r},Q_{s}^{\ast}\}&=&P^{\mu}(C\Gamma_{\mu})_{rs}+Z_{O}^{\mu \nu}(C\Gamma_{\mu \nu})_{rs}, \nonumber \\
\{Q_{r},\widehat{Q}_{s}^{\ast}\}&=&P^{\mu}(C\Gamma_{\mu})_{rs}+Z_{O}^{\mu \nu}(C\Gamma_{\mu \nu})_{rs}. \nonumber
\end{eqnarray}
It is worth mentioning, as highlighted in \cite{Toppan1} the spinors can be related to octonionic-valued 4-component vectors. An octonionic formulation of M-superalgebra has also been developed in \cite{An}.
\section{Conclusions}
We present a general formalism for thermal Lie superalgebras  and thermal versions of some superalgebras: Poincaré superalgebra and M-superalgebra. A quasitriangular Hopf algebra associated to Lie Poincaré superalgebra was derived in order to obtain non-commutative coordinates through of Drinfeld twist. For the M-superalgebra, relations for the hat operator were presented in order to deduce the thermal version of the representations. An extension for octonionic superalgebra was also obtained. As perspectives, others twist operators can be proposed as well as a phenomenological interpretation of these results. The structure of quasitriangular Hopf algebra related to M-superalgebra is also a topic of interest and it is in progress. \ \
\\

\textbf{Conflict of interest} The authors declare that there is no conflict of interest.


\begin{thebibliography}{}


\bibitem{Wess}  J. Wess, , J. Bagger, \emph{Supersymmetry and Supergravity}, Princeton University Press (1983).

\bibitem{Freund} P. J. O. Freund,  \emph{Introduction to Supersymmetry}, Cambridge Monographs on Mathematical Physics, Cambridge University Press (1986).

\bibitem{Sohnius} M. F. Sohnius, ``Introducing Supersymmetry'', Phs. Rep.  \textbf{128}, 39-204 (1985).

\bibitem{Sal} V. Salnikov, ``Supersymmetrization: AKSZ and Beyond?'', Russ. J. Math. Phys., \textbf{27}, 517-534 (2020).

\bibitem{Kac} V. Kac, ``Lie superalgebras'', Advances in Math. \textbf{26}, 8-96 (1977).

\bibitem{Traubenberg} M. R. Traubenberg, ``Clifford Algebras in Physics'', Advances in Applied Clifford Algebras \textbf{19} 868-908 (2009).

\bibitem{Kaku} M. Kaku, \emph{Introduction to Superstring and M-Theory}, 2nd edn, Springer Verlag, New York (1999).

\bibitem{Toppan1} F. Toppan, J. Lukierski, ``Octonionic superconformal M-algebra'', Int. J. M. Phys. A \textbf{18}, 2135-2141 (2003).

\bibitem{Toppan2} F. Toppan, J. Lukierski, ``Octonionic M-theory and D=11 generalized conformal and superconformal algebras'', Phys. Lett. B \textbf{567}, 125-132 (2003).

\bibitem{Townsend} P. Townsend, p-Brane Democracy, arXiv: hep-th/9507048.

\bibitem{Holten} J. V. Holten, J. V.,  A. V. Proyen,  ``N=1 Supersymmetry Algebras In D=2, D=3, D=4 Mod-8'', J. Phys. A, \textbf{15}, 3763-3784 (1982).


\bibitem{Nedel} D. Nedel, M. C. Abdalla, A. L. Gadelha, ``Superstring in a pp-wave background at finite temperature TFD approach'', Phys. Lett. B \textbf{598}, 121-131 (2004).

\bibitem{Khanna} F. C. Khanna, A. P. C. Malbouisson, J. M. C.   Malbouisson, A. E. Santana,  \emph{Thermal Quantum Field Theory: Algebraic Aspects and Applications}, World Scientific Publishing Company (2009).

\bibitem{Matos}   A. E. Santana, A. Matos Neto, J. D. M.  Vianna, F. C. Khanna,  ``w*-Algebra, Poincaré Group and Quantum Kinetic Theory'', International J. of Theor. Phys. \textbf{38}, 641-651 (1999).

\bibitem{Tak}  M. Takesaki, M. Tomita, \emph{Theory of Modular Hilbert Algebras and its Applications}, Springer-Verlag, Berlin (1970).

\bibitem{Ojima}  I. Ojima, ``Gauge Fields at Finite Temperatures - ``Thermo Field Dynamics'' and the KMS condition and Their Extension to Gauge Theories'', Annals of Physics \textbf{137}, 1-32 (1981).

\bibitem{Bratteli} O. Bratelli,  D. W. Robinson, \emph{Operators algebras and Quantum Statistical Mechanics}, Vols I and II,  Springer-Verlag, Berlin (1997).

\bibitem{Santana}  A. E. Santana, F. C., Khanna, ``Lie groups and thermal field theory'', Phys. Lett. A \textbf{203} 68-72 (1995).

\bibitem{Majid} S. Majid, \emph{Foundations of Quantum Group Theory}, Cambridge University Press, Camridge (1995).

\bibitem{Pre} V. Chari, A. Pressley, \emph{A Guide to Quantum Groups}, Cambridge University Press, Cambridge (1994).

\bibitem{Burban} I. M. Burban, ``(p, q)-Analog of Two-Dimensional Conformal Field Theory. The Ward Identities and Correlation Functions'', Journal of Nonlinear Mathematical Physics \textbf{2} 114-119 (1995).

\bibitem{Resh}  N. Reshetikhin, ``Multiparameter Quantum Groups and Twisted Quasitriangular Hopf Algebra'', Lett. Math. Phys. \textbf{20}, 331 (1990).

\bibitem{Kobayashi}  Y. Kobayashi, S.,  Sasaki, ``Lorentz invariant and supersymmetric interpretation of noncommutative quantum field theory'', Int. J. Mod. Phys. A \textbf{20} 7175-7178 (2005).

\bibitem{Nair} V. P. Nair, ``Thermofield Dynamics and Gravity'', Phys. Rev. D \textbf{92}, 104009-1-104009-15 (2015) .

\bibitem{Taylor} W. Taylor IV, ``M(atrix) theory: matrix quantum mechanics as a fundamental theory'', Rev. Mod. Phys. \textbf{73}, 419-461 (2001).

\bibitem{Kirsten} H. J. W. Müller, A. Weidmann, \emph{Supersymmetry: An Introduction with Conceptual and Calculational Details}, World Scientific, Singapore (1987).

\bibitem{An} A. Anastasiou, L. Borsten, M. J. Duff, L. J. Hughes, S. Nagy, ``An octonionic formulation of the M-theory algebra'', Journal of High Energy Physics \textbf{22}, 1-9, (2014).



\end{thebibliography}
\end{document}